\pdfoutput=1
\documentclass[amsmath,secnumarabic,floatfix,amssymb,nofootinbib,nobibnotes,letterpaper,11pt,tightenlines]{revtex4-2}

\usepackage{geometry}
\usepackage{amssymb}
\usepackage{latexsym}
\usepackage{amsmath}
\usepackage{amscd}
\usepackage{amsthm}
\usepackage{txfonts}
\usepackage{graphicx}
\usepackage[dvipsnames]{xcolor}
\usepackage[percent]{overpic}
\usepackage{booktabs}
\usepackage{units}

\usepackage{algorithm}
\usepackage{algorithmicx}
\usepackage{algpseudocode}

\usepackage{enumitem}
\setlist{nosep}
\def\figdir{Pictures/}
\graphicspath{\figdir}

\usepackage{aliascnt} %alias counter
\newcommand{\newautotheorem}[3] %{environment name}{counter}{displayed name} - Sets the autorefname of enronment so that \autoref from the \hyperref-Packet works correctly.
{
\newaliascnt{#1}{#2}
\newtheorem{#1}[#1]{#3}
\aliascntresetthe{#1}
\expandafter\def\csname #1autorefname\endcsname{%
#3%
}%
}

\newtheorem{theorem}{Theorem}

\newautotheorem{lemma}{theorem}{Lemma}
\newautotheorem{proposition}{theorem}{Proposition}
\newautotheorem{corollary}{theorem}{Corollary}

\theoremstyle{definition}
\newautotheorem{definition}{theorem}{Definition}

\newautotheorem{conjecture}{theorem}{Conjecture}
\newautotheorem{remark}{theorem}{Remark}
\newautotheorem{claim}{theorem}{Claim}

\AtBeginEnvironment{thebibliography}{\interlinepenalty=10000}

\newcommand{\R}{\mathbb{R}}

\newcommand{\dt}{\,\mathrm{d}t}

\newcommand{\Pol}{\operatorname{Pol}}

\newcommand{\PolHat}{\widehat{\Pol}(n)}
\newcommand{\APol}{\widehat{\operatorname{APol}}}

\newcommand{\Vol}{\operatorname{Vol}}

\newcommand{\SO}{\operatorname{SO}}

\newcommand{\sinc}{\operatorname{sinc}}

\newcommand{\newmethod}{progressive action-angle method}
\newcommand{\ceq}{\coloneqq}

\let\mgp=\marginpar \marginparwidth18mm \marginparsep1mm
\def\marginpar#1{\mgp{\raggedright\tiny #1}}

\let\lbl=\label
\def\label#1{\lbl{#1}\ifinner\else\marginpar{\ref{#1} #1}\ignorespaces\fi}

\usepackage[
	backref  = false%
	,pagebackref = false%
	,bookmarks  = true%
	,bookmarksdepth=3%
]{hyperref}
\usepackage[all]{hypcap}  

\begin{document}
\title{A faster direct sampling algorithm for equilateral closed polygons and the probability of knotting}
\author{Jason Cantarella}
\altaffiliation{Mathematics Department, University of Georgia, Athens, GA, USA}
\noaffiliation
\author{Henrik Schumacher}
\altaffiliation{Mathematics Department, University of Georgia, Athens, GA, USA}
\noaffiliation
\author{Clayton Shonkwiler}
\altaffiliation{Department of Mathematics, Colorado State University, Fort Collins, CO, USA}
\email{clayton.shonkwiler@colostate.edu}
\noaffiliation


\begin{abstract}
We present a faster direct sampling algorithm for random equilateral closed polygons in three-dimensional space. This method improves on the moment polytope sampling algorithm of Cantarella, Duplantier, Shonkwiler, and Uehara~\cite{Cantarella:2016bt} and has (expected) time per sample quadratic in the number of edges in the polygon. We use our new sampling method and a new code for computing invariants based on the Alexander polynomial to investigate the probability of finding unknots among equilateral closed polygons.
\end{abstract}
\date{\today}
\maketitle

\section{Introduction}

The tendency of long, flexible structures to become entangled is familiar to anyone who has managed an electrical extension cord, carried headphones in their pockets, or brushed a child's hair. This phenomenon has been studied for several decades in statistical physics~(cf.~\cite{Orlandini:2007kn}). A~standard ensemble for these investigations is the space of equilateral random polygons. These are polygonal walks in 3-space forming closed loops and consisting of unit-length steps. The closure condition imposes subtle global correlations between edge directions, which means it is not obvious how to generate random equilateral polygons. Indeed, algorithms have been proposed for at least 4~decades \cite{Anonymous:2010p2603,Cantarella:2013wl,Klenin:1988dt,MR95g:57016,Moore:2004ds,Varela:2009cda,Vologodskii:1979ik,Diao:2011ie,Diao:wt,Diao:2012dza,Moore:2005fh}, though most are numerically unstable or have not been proved to sample from the correct probability distribution.

In previous work~\cite{Cantarella:2016bt}, we introduced the \emph{action-angle method} (AAM), which is a numerically stable and provably correct algorithm for generating random equilateral $n$-gons in $\R^3$ based on rejection sampling the hypercube. The action-angle method is the fastest extant method: it produces samples in expected time $\Theta(n^{5/2})$. Xiong et al.~\cite{Xiong2021} used large-scale computing resources to sample $1.6 \times 10^9$ polygons using AAM, classifying their knot type with three invariants based on the Alexander polynomial. These invariants can be computed extremely quickly using sparse matrix methods (we estimate $O(n^{1.18})$, see \autoref{fig:timing}) and AAM is faster than invariant computation only for $n < 280$. In this paper, we give an improved algorithm, the~\emph{progressive} action angle method (PAAM), which we prove produces samples in $\Theta(n^2)$ time. Empirically, PAAM is faster than AAM for $n > 20$ and faster than invariant computation for $n < 1850$. 

One of the most interesting findings in~\cite{Xiong2021} concerns the probability $P_{0_1}(n)$ of finding an unknotted equilateral polygon among random polygons. Diao~\cite{Diao1995} proved that $P_{0_1}(n)$ is at most exponential in $n$, and it had been widely assumed that (up to subdominant terms) that $P_{0_1}(n) \simeq C_{0_1} \exp(-n/n_{0_1})$ for some constants $C_{0_1}$ and $n_{0_1}$~\cite{MillettRawdonUniversal}. (To be precise, expressions of the form $P_{0_1}(n) \simeq C_{0_1} n^{v_0} \exp(-n/n_K)$ appeared early on in the literature (cf. \cite[eq.\ (1.4)]{des_cloizeaux_topological_1979},~\cite[eq.\ (2)]{michels_probability_1982}, \cite[eq.\ (1.3)]{deguchi_statistical_1994}) but the power $v_0$ was thought to be close to zero based on the data sets available at the time.) This made the unknot an anomaly, since every other knot type studied fit to the general form $P_K(n) \simeq C_K n^{v_K} \exp(-n/n_K)$. Xiong et al.~\cite{Xiong2021} presents strong evidence that the probability of a given knot type is in the form 
\[
P_K(n) \simeq C_K n^{v_0 + n_p(K)} \exp(-n/n_0) [1 + \beta_K n^{-\frac{1}{2}} + \gamma_K n^{-1} ],
\]
where $v_0 \simeq -0.19 \pm 0.001$ and $n_0 \simeq 259.3 \pm 0.2$ are constants that do~\emph{not} depend on the knot type, and $n_p(K)$ is the number of prime components of the knot type. Deguchi and Uehara~\cite{DeguchiUehara2020} proposed a very similar model (with a different finite-size correction) and, as Xiong et al.\ point out, a model of a similar form was proposed for self-avoiding polygons on lattices by Orlandini et al.~\cite{orlandini_asymptotics_1999}. 

In this model, the unknot is better thought of as ``a composite knot with zero components'' instead of as an anomalous knot type. If this is true, then 
\begin{equation}
P_{0_1}(n) \simeq C n^{-0.19} \exp(-n/259.3) [1 + \beta n^{-\frac{1}{2}} + \gamma n^{-1} ]\label{eq:xdwform}
\end{equation}
with free parameters $C$, $\beta$, and $\gamma$ instead of 
\begin{equation}
P_{0_1}(n) \simeq \exp(-n/N) [1 + \beta n^{-\frac{1}{2}} + \gamma n^{-1} ] \label{eq:pureexpform}
\end{equation}
with free parameters $N$, $\beta$ and $\gamma$. Xiong et al.\ find that~\eqref{eq:xdwform} gives a good fit, but \eqref{eq:pureexpform} does not. 

Using PAAM, we were able to generate data for unknot probabilities for $n$ up to $3043$ (see~\autoref{tab:unknot probs}) using a desktop computer. We confirm~Xiong et al.'s finding that~\eqref{eq:xdwform} fits the data very well and we obtain similar estimates for the free parameters. However, unlike those authors, we find that~\eqref{eq:pureexpform} fits the data equally well, though with very different values for $\beta$ and $\gamma$. We conclude that more will be required to distinguish between these models. 

\section{The Progressive Action-Angle Method}

For an $n$-gon in $\R^3$, let $v_1, \dots , v_n \in \R^3$ be the coordinates of its vertices, and let $e_1, \dots , e_n$ be the edge vectors, meaning that $e_i = v_{i+1} - v_i$ for $i=1,\dots , n-1$ and $e_n = v_1 - v_n$. We will assume throughout that our polygons are equilateral, so that $|e_i| = 1$ for all $i$; equivalently, $e_1, \dots , e_n \in S^2$, the unit sphere in $\R^3$. The space $\Pol(n)$ consists of sets of edge vectors in $(S^2)^n$ which obey the closure condition $\sum_{i=1}^n e_i = 0$. 
One can show that the set 
\begin{align*}
	\Pol(n)^\times \ceq  \big\{\vec{e} \in (S^2)^n : \text{$\textstyle\sum_{i=1}^n e_i = 0$ and for all $i \neq j$: $e_i \neq e_j$ } \big\}
\end{align*}
is a $(2n-3)$-dimensional submanifold of $(S^2)^n$ and that the $(2n-3)$-dimensional Hausdorff measure of $\Pol(n) \setminus \Pol(n)^\times$ vanishes.
In this sense $\Pol(n)$ is almost everywhere a submanifold of $(S^2)^n$. 
We~may give it the submanifold metric and corresponding volume; it is equivalent to taking the $(2n-3)$-dimensional Hausdorff measure on $\Pol(n)$ with respect to the metric on $(S^2)^n$.

We focus on the quotient space ${\PolHat = \Pol(n)/\SO(3)}$. This space has a Riemannian metric---defined by the condition that the quotient map $\Pol(n) \to \PolHat$ is a Riemannian submersion---and hence a natural probability measure after normalizing the Riemannian volume form.

\begin{figure}[t]
\begin{overpic}[height=1in]{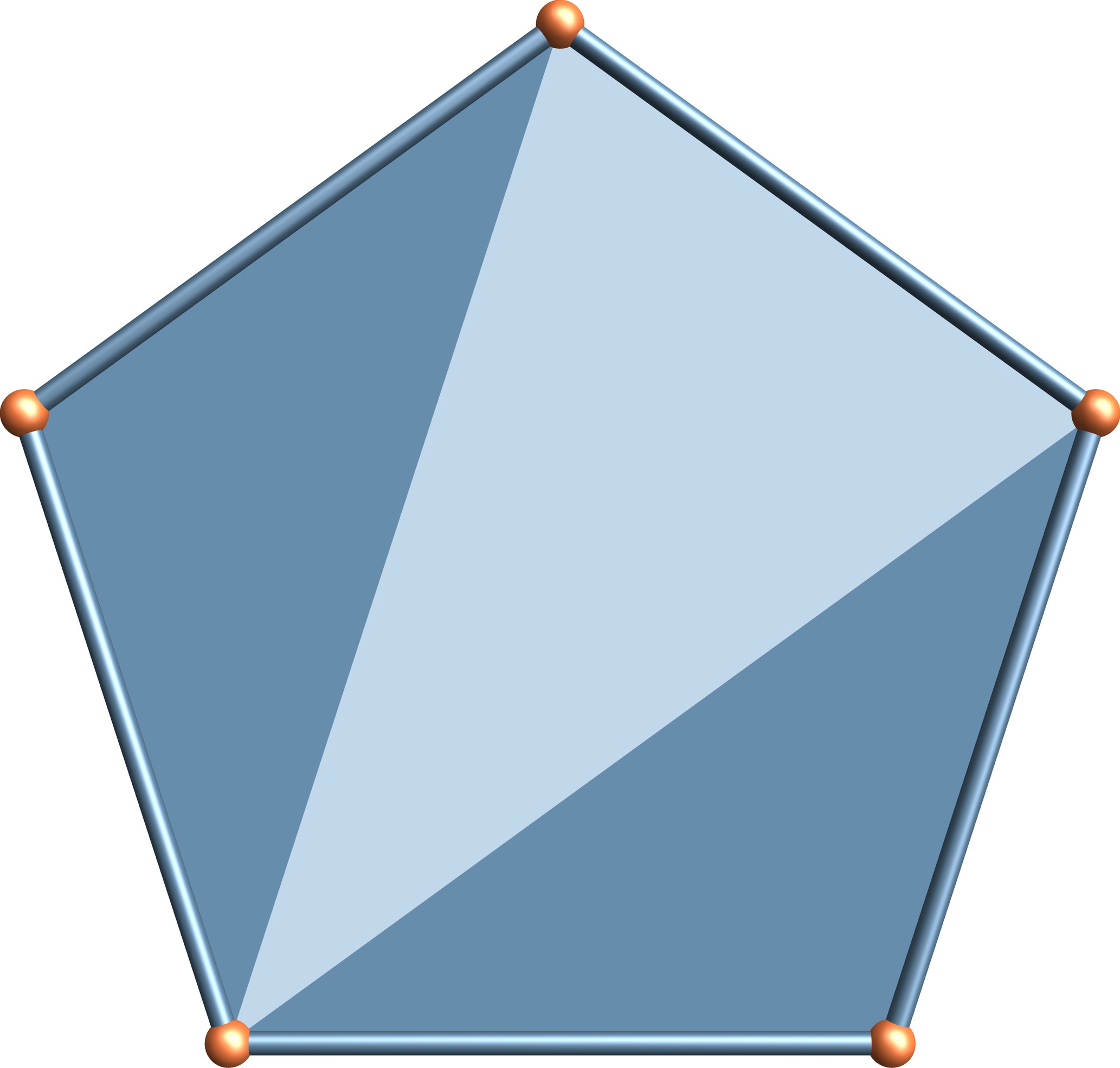}
\put(45,98){$v_1$}
\put(-13,57){$v_2$}
\put(17,-9){$v_3$}
\put(77,-9){$v_4$}
\put(102,57){$v_5$}
\end{overpic}
\hfill
\begin{overpic}[height=1in]{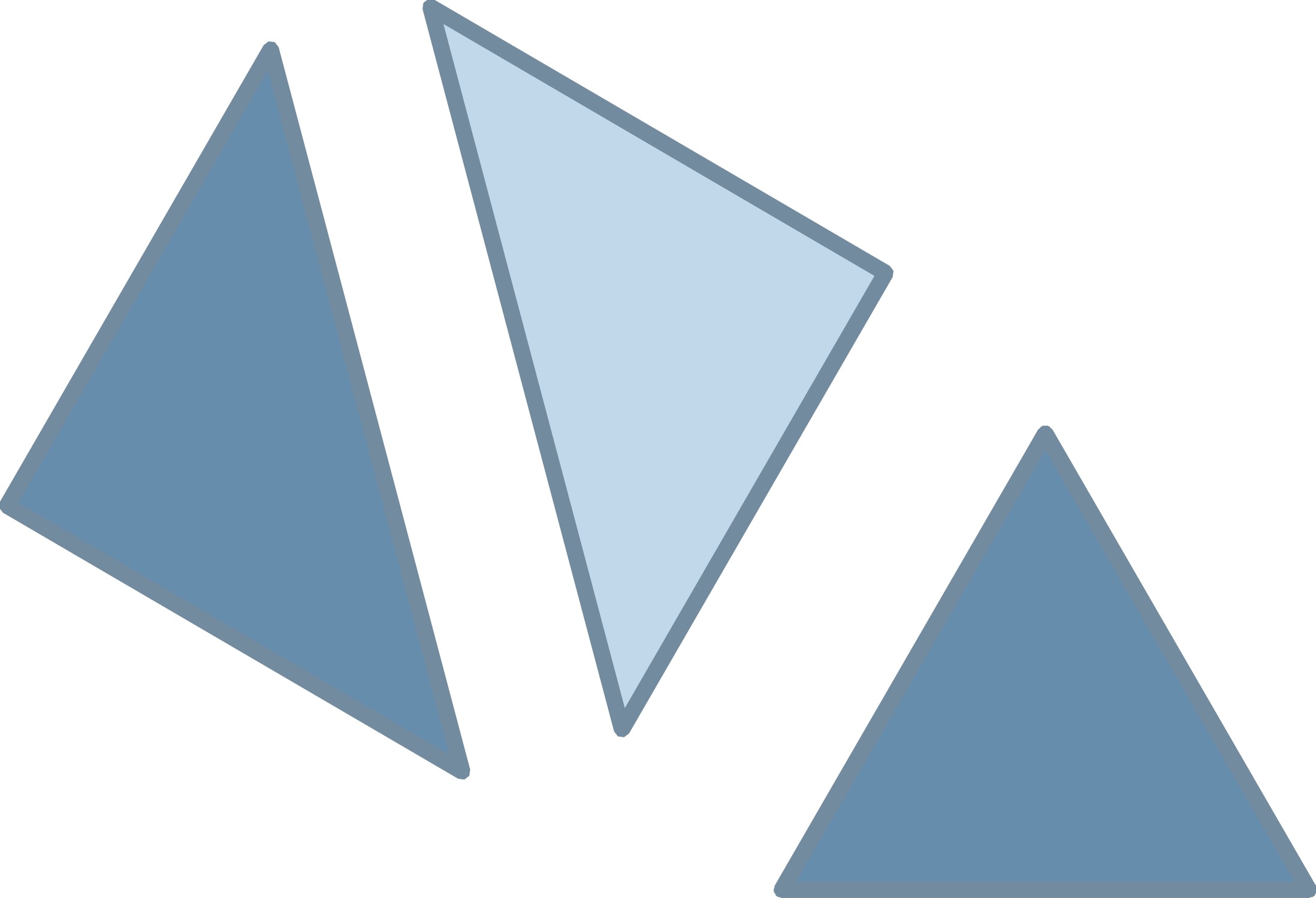}
\put(30,35){$d_1$}
\put(59,23){$d_2$}
\end{overpic}
\hfill
\begin{overpic}[height=1in]{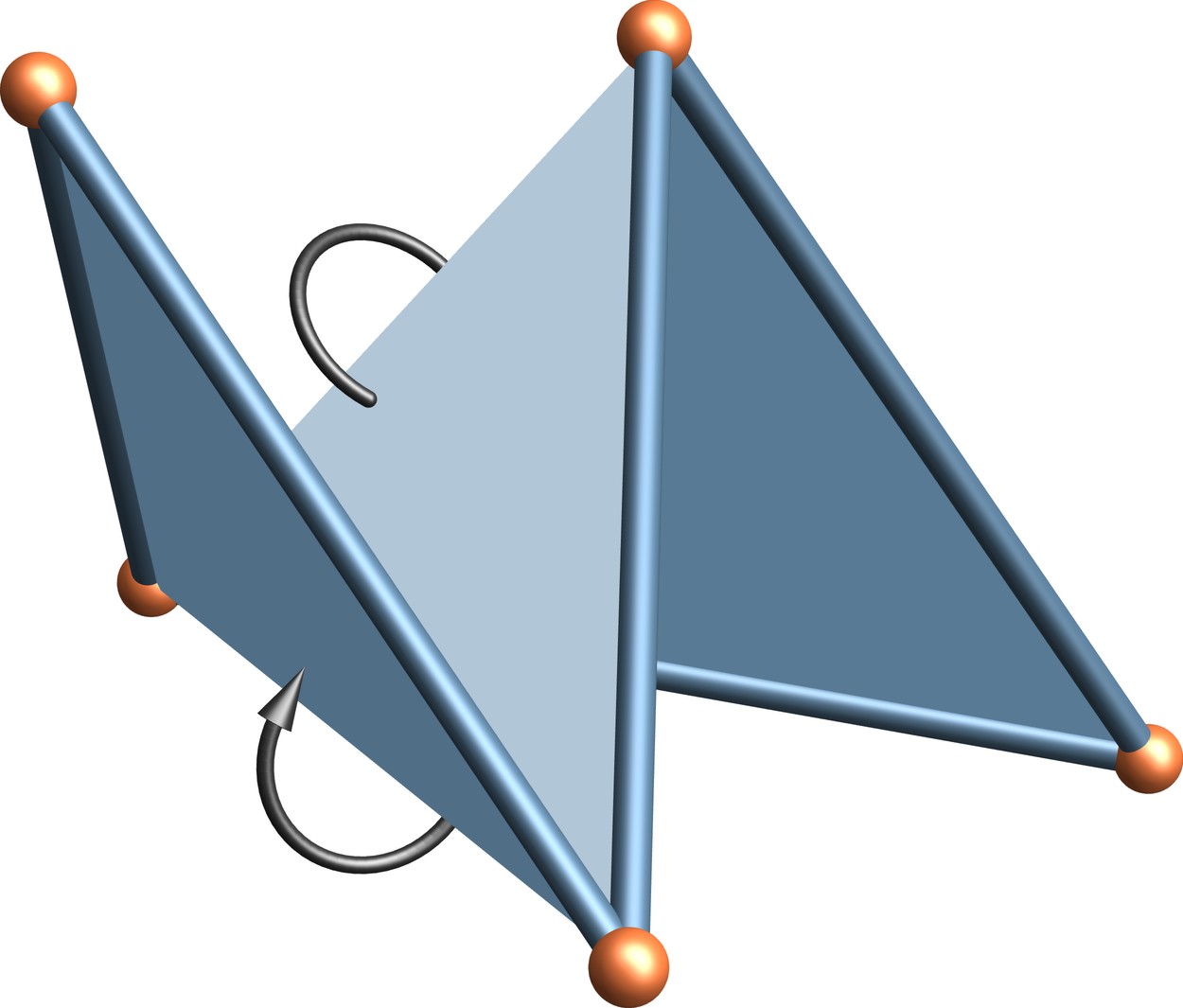}
\put(15,5){$\theta_1$}
\put(23,71){$\theta_2$}
\end{overpic}
\hfill
\begin{overpic}[height=1in]{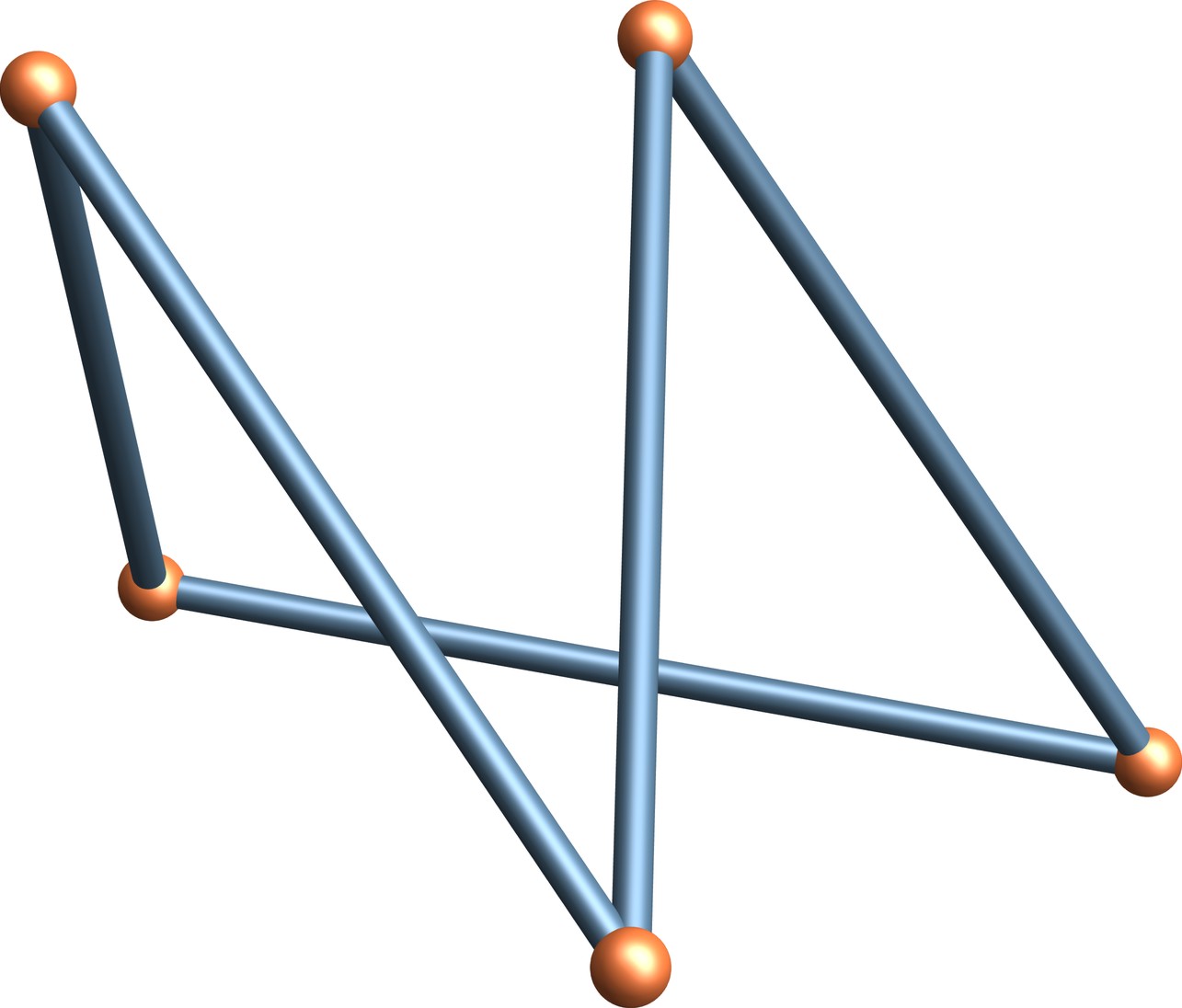}
\end{overpic}
\caption{Constructing an equilateral pentagon from diagonals and dihedrals. The far left shows the fan triangulation of an abstract pentagon. Given diagonal lengths $d_1$ and $d_2$ of the pentagon which obey the triangle inequalities, build the three triangles in the triangulation from their side lengths (middle left). Given dihedral angles $\theta_1$ and $\theta_2$, embed these triangles as a piecewise-linear surface in space (middle right). The far right shows the final space polygon, which is the (solid) boundary of this triangulated surface.} 
\label{fig:assembly} 
\end{figure}
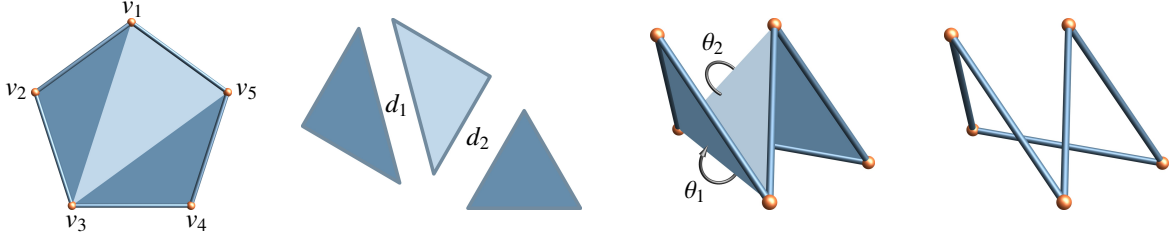

Now we introduce some new coordinates on this space. Connecting the vertices $v_3, \dots, v_{n-1}$ to $v_1$, as in \autoref{fig:assembly} (far left), produces a collection of $n-3$ triangles. The shape of the triangulated surface determined by these triangles (and hence also its boundary, which is the $n$-gon) is completely determined by the lengths $d_i$ of the diagonals joining $v_1$ and $v_{i+2}$ and the dihedral angles between triangles meeting at each diagonal. Hence, we can reconstruct the surface (and hence the polygon) up to orientation from the data $d_1, \dots , d_{n-3}, \theta_1, \dots , \theta_{n-3}$, and so these give a system of coordinates for $\PolHat$. 

Indeed, as we have shown~\cite{Cantarella:2013wl}, these coordinates are natural from the symplectic geometry point of view: in that context, they are called \emph{action-angle coordinates}. Note that, while the dihedral angles can be chosen completely independently, the diagonal lengths cannot: they must obey the system of triangle inequalities
\begin{equation}
0 \leq d_1 \leq 2 
\qquad 
\begin{matrix} 
1 \leq d_i + d_{i+1} \\
-1 \leq d_{i+1} - d_i \leq 1 
\end{matrix}
\qquad
0 \leq d_{n-3} \leq 2. 
\label{eq:fan polytope}
\end{equation}

Let $\mathcal{P}_n \subset [-1,1]^{n-3}$ be the polytope defined by the inequalities~\eqref{eq:fan polytope}. If $T^{n-3} = (S^1)^{n-3}$ is the $(n-3)$-dimensional torus realized as the product of unit circles, then the action-angle coordinates are defined on $\mathcal{P}_{n-3} \times T^{n-3}$, and we have previously shown that the standard probability measure on this space---that is, the one coming from the product of Lebesgue measure on $\mathcal{P}_n$ and the standard product measure on $T^{n-3}$---is measure-theoretically equivalent to $\PolHat$:

\begin{theorem}[{Cantarella--Shonkwiler~\cite{Cantarella:2013wl}}]\label{thm:measures}
	The reconstruction map $\alpha: \mathcal{P}_n \times T^{n-3} \to \PolHat$ defining action-angle coordinates (i.e., the procedure illustrated in \autoref{fig:assembly}) is measure-preserving.
\end{theorem}

Therefore, to sample points in $\PolHat$ (that is, equilateral $n$-gons), it suffices to sample $\vec{d}$ from Lebesgue measure on $\mathcal{P}_n$ and $\vec{\theta}$ uniformly from $T^{n-3}$. Of course, the only challenge is to produce the sample $\vec{d} \in \mathcal{P}_n$. 

In~\cite{Cantarella:2016bt}, we showed how to do this efficiently. The key observation is that the consecutive differences $s_i \ceq d_{i+1} - d_i$ lie in the hypercube $[-1,1]^{n-3}$. Therefore, we can generate points in $\mathcal{P}_n$ by rejection sampling: generate proposed differences $(s_0, \dots , s_{n-4})$ uniformly from $[-1,1]^{n-3}$, and simply check whether the proposed diagonal lengths $(d_1, \dots , d_{n-3})$ given by $d_{i+1} = d_i + s_i$ with $d_0 = |v_2-v_1|=1$ satisfy~\eqref{eq:fan polytope}. 
This is surprisingly efficient:

\begin{theorem}[Cantarella--Duplantier--Shonkwiler--Uehara~\cite{Cantarella:2016bt}]\label{thm:rejection sampling probability}
	The probability that a random point $(s_0, \dots, s_{n-4}) \in [-1,1]^{n-3}$ produces a valid collection of diagonal lengths $(d_1, \dots , d_{n-3}) \in \mathcal{P}_n$ is asymptotically equivalent to $\frac{6\sqrt{6}}{\sqrt{\pi}}\frac{1}{n^{3/2}}$ as $n \to \infty$.
\end{theorem}

In the above and throughought the rest of the paper, we say that $g(n)$ and $h(n)$ are \emph{asymptotically equivalent}, denoted $g(n) \sim h(n)$, if $\lim_{n \to \infty} \frac{g(n)}{h(n)} = 1$.

Since the time it takes to generate points in $[-1,1]^{n-3}$ is linear in $n$, rejection sampling the hypercube yields a valid point in $\mathcal{P}_n$ in expected time $\Theta(n^{5/2})$. The steps of generating dihedral angles and assembling the $n$-gon from $(d_1,\dots , d_{n-3})$ and $(\theta_1, \dots , \theta_{n-3})$ do not affect the time bound since they are both linear in $n$. Therefore, this gives a numerically stable algorithm for generating random equilateral $n$-gons in expected time $\Theta(n^{5/2})$.
We called it the \emph{action-angle method} (AAM).

We now give an improved version of this method. By fixing $d_0 = 1$ and generating proposed consecutive differences $s_i$ uniformly from $[-1,1]^{n-3}$, we guarantee that the inequalities $0 \leq d_1 \leq 2$ and $-1 \leq d_{i+1} - d_i \leq 1$ for $i=1,\dots , n-4$ are automatically satisfied. The final inequality $0 \leq d_{n-3} \leq 2$ can only be checked at the very end, but the inequalities $1 \leq d_i + d_{i+1}$ can be checked one at a time as each $s_i$ is generated, and we can abort and start over as soon as one of these inequalities fails. We will show below that the expected number of coordinates that get generated before failure is $\Theta(\!\sqrt{n})$, yielding an overall time bound of $\Theta(n^2)$. \autoref{alg:PAAM} gives an implementation of our approach, which we call the \emph{\newmethod} (PAAM). A~reference C implementation of this algorithm is included in \texttt{plCurve}~\cite{plcurve}, where it is now the default algorithm for producing random equilateral polygons in $\R^3$.

\section{Complexity}

Let $\mathcal{I}(n)$ be expected number of coordinates $s_i$ that are generated in the innermost loop before either failing one of the $d_i + d_{i+1} \geq 1$ inequalities (and resetting $i$ to $0$) or passing all of them (when $i=n-3$).
Since we know from \autoref{thm:rejection sampling probability} that the overall acceptance probability is $\sim \frac{6 \sqrt{6}}{\sqrt{\pi}} \frac{1}{n^{3/2}}$, the expected number of times we will have to re-start the innermost loop is $\Theta(n^{3/2})$. 
Multiplying these quantities yields $\Theta(n^{3/2} \mathcal{I}(n))$ for the expected time to produce a valid list of diagonals. The postprocessing steps of generating dihedral angles and assembling the polygon are both linear in~$n$, so they do not affect the overall time bound. 

\begin{algorithm}[H]
%\capstart
\begin{algorithmic}
\Procedure{ProgressiveActionAngleMethod}{$n$}\Comment{Generate closed equilateral $n$-gon}
\State $d_0 \gets 1$
\State $i \gets 0$
\Repeat
	\Repeat
		\State $s_{i} \gets $\Call{UniformRandom}{$[-1,1]$}
		\State $d_{i+1} \gets d_i \,+ s_{i}$
		\If{$d_i + d_{i+1} < 1$}
			\State $i \gets 0$
		\Else
			\State $i \gets i+1$
		\EndIf
	\Until{$i=n-3$}
\Until{$0 \leq d_{n-3} \leq 2$}
\State Sample $n-3$ i.i.d.\ dihedral angles $\theta_i$ uniformly from $[0,2\pi)$.
\State Reconstruct $P$ from diagonals $d_1, \dots, d_{n-3}$ and dihedrals $\theta_1, \dots, \theta_{n-3}$.
\EndProcedure
\end{algorithmic}
\caption{Progressive action-angle method}
\label{alg:PAAM}
\end{algorithm}

We now show
\begin{proposition}\label{prop:sum-asymptotics}
	$\mathcal{I}(n) \sim \sqrt{\frac{24n}{\pi}}.$
\end{proposition}

\begin{proof}
Let $p(k)$ be the probability that the proposed diagonal lengths generated by a random point in $\vec{s} \in [-1,1]^{n-3}$ satisfy each of the first $k$ inequalities $d_0 + d_1 \geq 1, \dots , d_{k-1} + d_k \geq 1$. In other words, $p(k)$ (thought of as a function of $k$) is the survival function for the distribution of the index of the first inequality which fails. For ease of notation, we also declare $p(0) \ceq 1$. By a standard integration by parts argument (see, for example,~\cite[Exercise~1.7.2]{durrett_probability_2019}), the expected value of a~nonnegative random variate is the integral of the survival function, so $\mathcal{I}(n) = \sum_{k=0}^{n-3} p(k)$. We will now show that for any nonnegative integer $k$, 
	\[
		p(k) = \frac{2}{\pi}\int_0^\infty \sinc^{k+1}(t) \dt,
	\]
	where $\sinc(t) = \frac{\sin(t)}{t}$ for $t \neq 0$ and $\sinc(0) = 1$ is the sinc function.
	
For each $k$, we denote by $\mathcal{S}_k \subset [-1,1]^k$ the subset of points $\vec{s}$
that satisfy the inequalities $d_0 + d_1 \geq 1, \dots , d_{k-1} + d_k \geq 1$. 
As a simple calculation confirms, the affine transformation from $\vec{s}$ to $\vec{d}$ given by $d_i = 1 + \sum_{j=0}^{i-1} s_j$ is volume-preserving. Let $\mathcal{Q}_k$ be the image of $\mathcal{S}_k$ under this map; i.e., $\mathcal{Q}_k$ is the polytope of $(d_1, \dots , d_k)$ satisfying the inequalities $d_0 + d_1 \geq 1, \dots , d_{k-1} + d_k \geq 1$ (again, recall that $d_0 = 1$, which implies in particular that all the $d_i$ are nonnegative). 
Since $p(k) = \frac{\Vol(\mathcal{S}_k)}{\Vol([-1,1]^k)} = \frac{\Vol(\mathcal{S}_k)}{2^k}$, to prove the above formula for $p(k)$ it suffices to prove that
	\[
		\Vol(\mathcal{Q}_k) = 2^k \frac{2}{\pi} \int_0^\infty \sinc^{k+1}(t) \dt.
	\]
	
	Notice that the defining inequalities for $\mathcal{Q}_k$ are precisely the inequalities~\eqref{eq:fan polytope} \emph{except} the last inequality $0 \leq d_k \leq 2$. This suggests that these inequalities may just be the diagonal lengths of an equilateral polygonal path in $\R^3$ which is not required to close up. 
	
	More precisely, let 
	\[
		\APol(k) \ceq \Big\{(e_1, \dots , e_{k+1}) \in (S^2)^{k+1} : z_1 + \dots + z_{k+1} = 0 \Big\}/\SO(2),
	\]
	where $e_i = (x_i,y_i,z_i)$, so that the defining condition says that the path starts and ends in the \mbox{$xy$-plane}, and $\SO(2)$ acts by simultaneously rotating all edges around the $z$-axis; this is the diagonal subgroup of the $T^{k+1}= \SO(2)^{k+1}$ action on $(S^2)^{k+1}$ which rotates edges around the $z$-axis. $\APol(k)$ is the space of \emph{abelian polygons} introduced by Hausmann and Knutson~\cite{Hausmann:1998vx}, and we will see that it admits two different effective, Hamiltonian $T^k$ actions, as shown in~\autoref{fig:actions}.

\begin{figure}[t]
\hfill
\begin{overpic}[height=0.225\textheight]{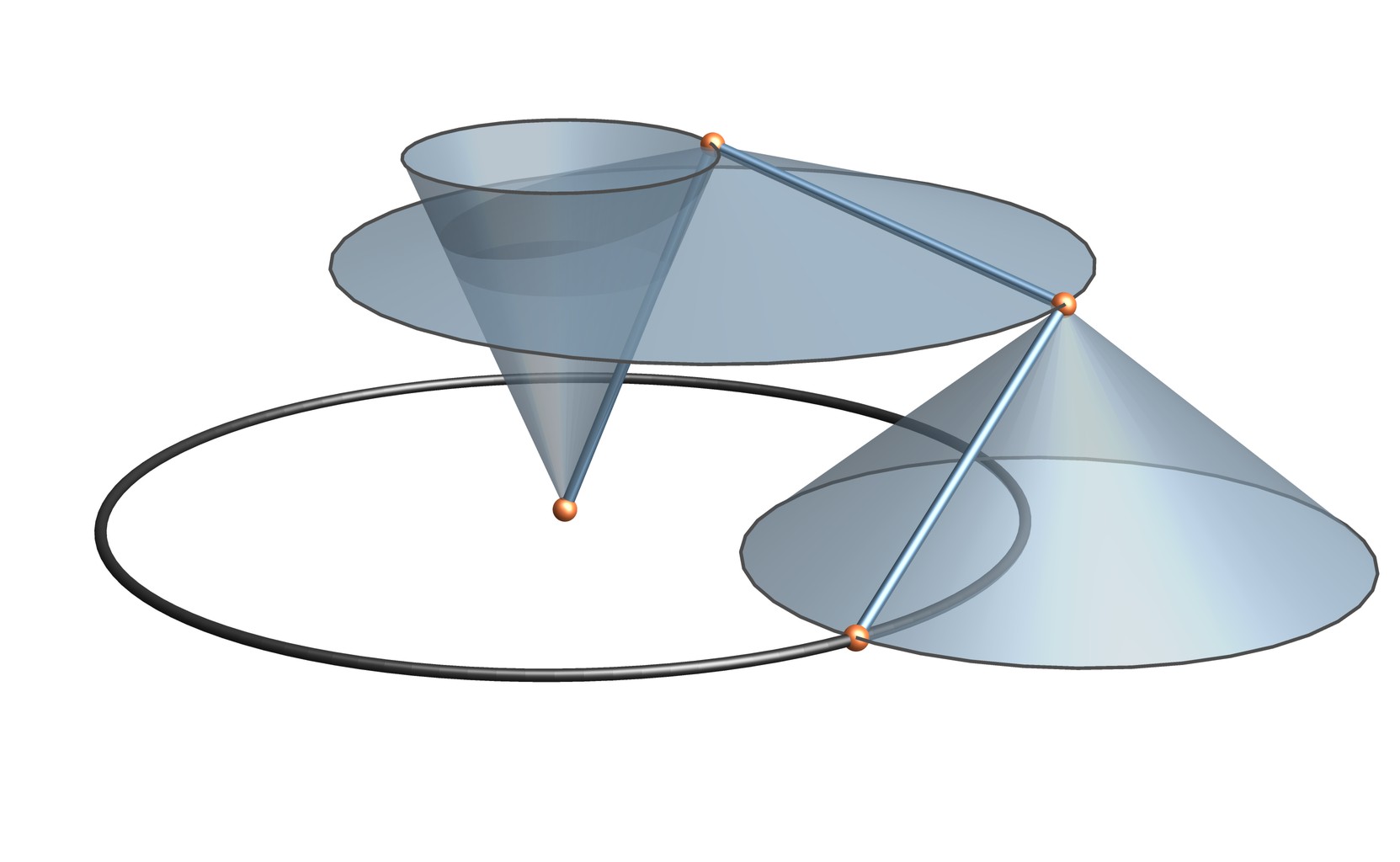}
	\put(34,23){$v_1$}
	\put(50,55){$v_2$}
	\put(80,40){$v_3$}
	\put(60,10){$v_4$}
\end{overpic}%
\hfill
%\hspace{0.5in}%
%\raisebox{0.2in}{
\begin{overpic}[height=0.225\textheight]{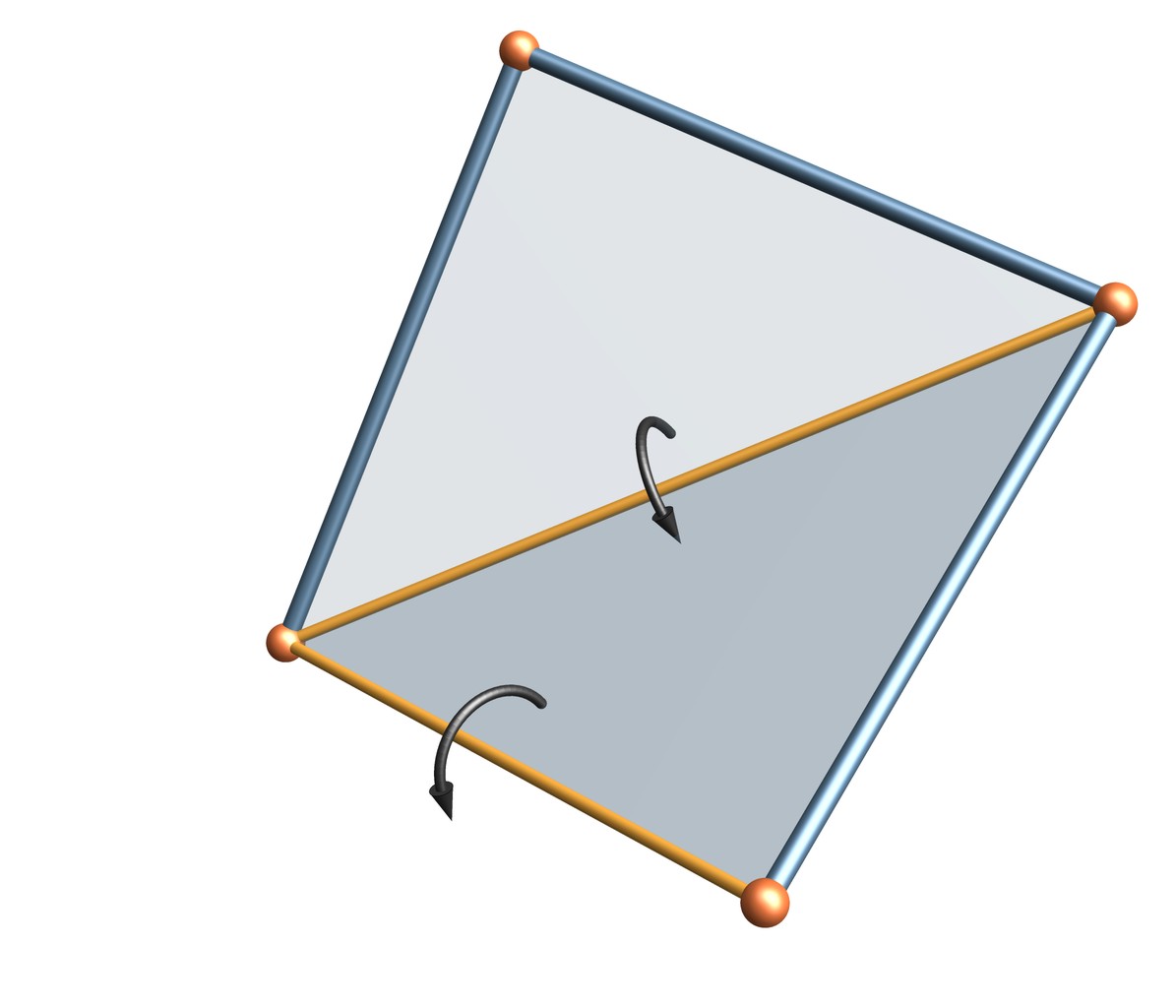}
	\put(15,25){$v_1$}
	\put(32,80){$v_2$}
	\put(92,67){$v_3$}
	\put(72,5){$v_4$}
\end{overpic}%
\hfill\hphantom{.}
%}
\caption{Here we see both torus actions on the space $\APol(2)$. On the left, we may rotate each of the three edges (independently) around the $z$-axis, sweeping out three cones. Then we identify configurations which are the same under the diagonal circle action rotating the entire configuration around the $z$-axis (indicated by the dark circle in the $xy$-plane). On the right, we may rotate the first two edges around the diagonal joining vertices $v_1$ and $v_3$ or rotate the entire polygon around the diagonal joining $v_1$ and $v_4$. In each case, this is a Hamiltonian $2$-torus action on $\APol(2)$.} 
\label{fig:actions} 
\end{figure}
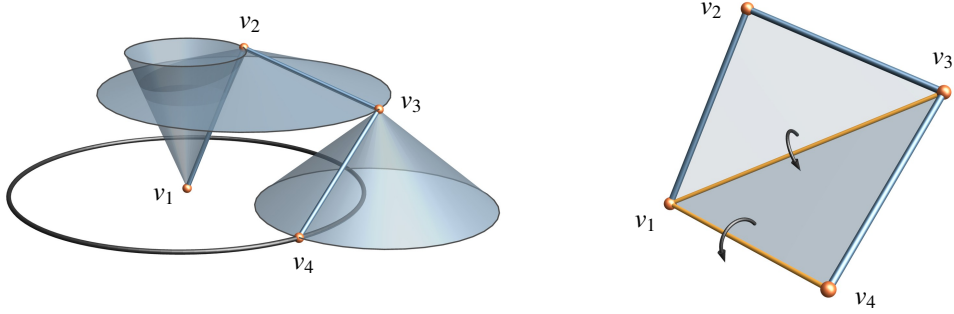

	The first is the residual $T^{k+1}/\SO(2) \simeq T^k$ action above. The moment map for this action simply records the $z$-coordinates of the edges; since the defining equation $z_1 + \dots + z_{k+1} = 0$ implies that $z_{k+1}$ is determined by the remaining $z_i$'s, the last coordinate can be dropped and we can think of the moment map as recording the vector $(z_1, \dots , z_k)$. Let $\mathcal{H}_k$ be the image of this map; that is, the moment polytope for this torus action. Of course, $-1 \leq z_i \leq 1$ for all $i$, and, since $-1 \leq z_{k+1} \leq 1$, we see that the defining inequalities of $\mathcal{H}_k$ are
	\[
		-1 \leq z_i \leq 1 \text{ for all } i =1, \dots, k \quad \text{and} \quad -1 \leq z_1 + \dots + z_k \leq 1.
	\]
	In other words, $\mathcal{H}_k$ is the central slab of the hypercube $[-1,1]^k$ of points whose coordinates sum to between $-1$ and $1$. This is a well-studied polytope, and its volume has been known at least since P\'olya~\cite{polya_berechnung_1913} to be
	\[
		\Vol(\mathcal{H}_k) = 2^k \frac{2}{\pi}\int_0^\infty \sinc^{k+1}(t) \dt
	\]
	(see also Borwein, Borwein, and Mares~\cite{Borwein:2002fj} for generalizations of the above formula).
	
	On the other hand, we get a $T^k$ action on $\APol(k)$ which is analogous to the bending flows on $\PolHat$ described above and in more detail in~\cite{Cantarella:2013wl}. Specifically, the $i$th $\SO(2)$ factor acts by rotating the first $i+1$ edges of the polygonal arm around the $i$th diagonal, which is the axis through $v_1$ and $v_{i+2}$. Just as in the case of $\PolHat$, the moment map records the lengths $d_1, \dots , d_k$ of the diagonals, and the image of the moment map is precisely $\mathcal{Q}_k$.
	
	But now we've realized $\APol(k)$ as a toric symplectic manifold in 2 ways, with moment polytopes $\mathcal{H}_k$ and $\mathcal{Q}_k$. Since the Duistermaat--Heckman theorem~\cite{Duistermaat:1982hq} implies that the volume of $\APol(k)$ must be the product of the volume $(2\pi)^k$ of the torus $T^k$ and the volume of either of its moment polytopes, it follows that
	\[
		\Vol(\mathcal{Q}_k) = \Vol(\mathcal{H}_k) = 2^k \frac{2}{\pi}\int_0^\infty \sinc^{k+1}(t) \dt,
	\]
	as desired. We can now estimate the integral using classical methods~\cite[pp.~172--173]{laplace_theorie_1820} (see also~\cite{medhurst_evaluation_1965}) to get $p(k) = \sqrt{\frac{6}{\pi k}} + O\left(\frac{1}{k^{3/2}}\right)$
and the result follows.
\end{proof}

%Now we can easily get the desired asymptotic expression for $\mathcal{I}(n)$:
%
%
%
%\begin{proposition}\label{prop:sum-asymptotics}
%	$\mathcal{I}(n) \sim \sqrt{\frac{24n}{\pi}}.$
%\end{proposition}
%
%\begin{proof}
%	By \autoref{prop:sinc-integral-estimate}, $p(k) = \sqrt{\frac{6}{\pi k}} + O\left(\frac{1}{k^{3/2}}\right)$, so, from~\eqref{eq:fn},
%	\[
%		\mathcal{I}(n) = \sum_{k=0}^{n-3} p(k) = 1 + \sum_{k=1}^{n-3} \left[ \sqrt{\frac{6}{\pi k}} + O\left(\frac{1}{k^{3/2}}\right)\right] = 1+\sum_{k=1}^{n-3} \sqrt{\frac{6}{\pi k}} + \sum_{k=1}^{n-3} O\left(\frac{1}{k^{3/2}}\right).
%	\]
%	The first sum is asymptotic to $\sqrt{\frac{24n}{\pi}}$ by the integral test. The rest of this expression has order $O(1)$ since the second sum is a partial sum of a convergent series, and the result follows.
%\end{proof}

Therefore, $\Theta(n^{3/2} \mathcal{I}(n)) = \Theta(n^2)$ and we have proved a sharp complexity bound on the \newmethod:

\begin{theorem}\label{thm:main bound}
	The \newmethod\ generates uniform random samples of closed, equilateral $n$-gons in expected time $\Theta(n^2)$.
\end{theorem}

\section{Experiments}
\label{sec:experiments}

As in~\cite{Xiong2021}, we classified each random polygon $P$ as knotted or unknotted based on the three invariants 
$\Delta_2(P)$, $\Delta_3(P)$, and $\Delta_4(P)$, which are absolute values of the Alexander polynomial at roots of unity. The $\Delta_i$ are determinants of principal minors of an $m \times m$ Alexander matrix $A(P)$~\cite{AdamsKnotBook}, where $m$ is the number of crossings in a projection of $P$ to a plane. The matrix $A(P)$ is usually large, since the expected value of $m$ is $\sim \frac{3}{16} n \log n$~\cite{DiaoDobayKusnerMillett2003}, but sparse, with at most $3$ nonzero entries in each row. 
We computed $\det (A^T A)$ by using a sparse Cholesky factorization: 
% First we determined a fill-in reducing reordering of the matrix $A^TA$ with the routine \texttt{METIS\textunderscore{}NodeND} from the \emph{Metis} library~\cite{Metis}. 
First we determined a fill-in reducing reordering of the matrix $A^TA$ with the \emph{Approximate Minimum Degree} algorithm \cite{10.1145/1024074.1024081,doi:10.1137/S0895479894278952}
(routine \texttt{amd\textunderscore{}order} from the \emph{SuiteSparse} library). 
Then we factorized the reordered matrix and extracted the diagonal entries of the factor matrix with our own code~\cite{Tensors}. Finally we computed the determinant as the product of these diagonal entries.

Our first experiment compared the empirical performance of AAM, PAAM, and invariant computation on an Apple M1 Max laptop ($8$~cores in parallel). Fitting the data shown in~\autoref{fig:timing} confirmed the expected runtime of $O(n^{5/2})$ for AAM and $O(n^2)$ for PAAM and revealed that invariant computation seems to scale as $O(n^{1.18})$.

\begin{figure}[t]
	\centering
		\includegraphics[height=2.5in]{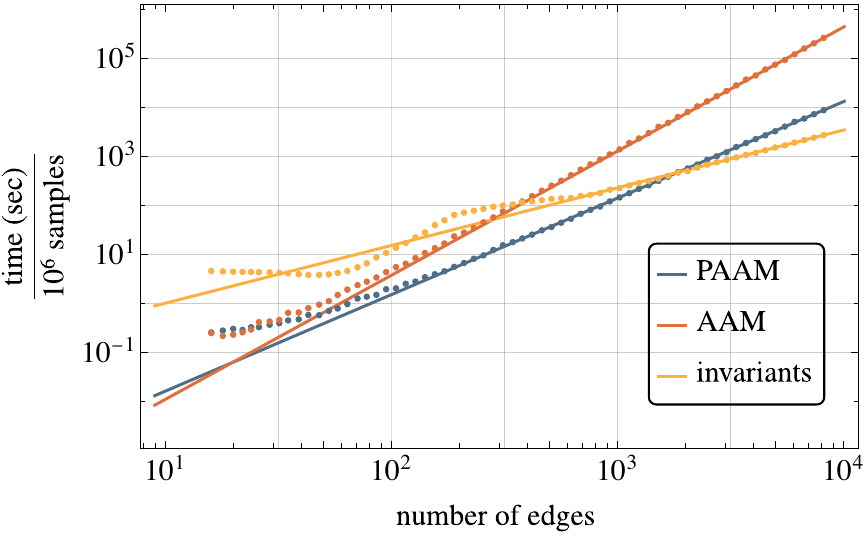}
	\caption{This plot compares time per million samples to $n$ for random equilateral $n$-gons with $n=\operatorname{round}(2^{k/7})$ from $n=16$ to $8192$ using the action-angle method (AAM) and the progressive action-angle method (PAAM). The time needed to generate samples scales as predicted by \autoref{thm:rejection sampling probability} and~\autoref{thm:main bound}. We compare these to the time required for the computation of all three invariants $\Delta_2$, $\Delta_3$ and $\Delta_4$. Invariant computation shows some small-$n$ effects. The most important one is that it is faster to use a dense matrix code to compute the determinants for $n < 128$. Dense matrix determinants scale as $O(n^3)$, so this portion of the graph is steeper. For $n < 64$ the determinant computation is so fast that (constant time) computational overhead is the dominant factor in the computation, so this portion of the timing curve is flat. AAM is faster than invariant computation for $n < 280$ while PAAM is faster than invariant computation for $n < 1850$.}
	\label{fig:timing}
\end{figure}

Our second experiment recomputed unknot probabilities using our invariant code. Since unknotted $n$-edge polygons are very uncommon for large $n$, we used inverse sampling~\cite{Lui2004}. We sampled polygons with $n = \operatorname{round}(2^{k/7})$ edges for integer $k \in [28,81]$, covering the range $n=16$ to $n=3043$. For each $n$, we sampled until observing $600$ polygons classified as unknots, computing between $608$ and $\sim 10^8$ samples to do so. The computations took a total of 58 hours on an Apple M1 Ultra personal computer with $16$~CPU cores. Using Bennett's approximation for the negative binomial distribution~\cite{Bennett1981}, we computed $95\%$ confidence intervals for our unknot probabilities of relative size about $8\%$ of the probability (see~\autoref{tab:unknot probs}). 

We then fit the data to models of the form~\eqref{eq:xdwform} and~\eqref{eq:pureexpform}. As shown in~\autoref{fig:fits}, model~\eqref{eq:xdwform} provides an excellent ($R^2 > 0.9999$) fit to the data, with best-fit parameters $C \simeq 3.63 \pm 0.09$, $\beta \simeq 3.8 \pm 0.4$, and $\gamma \simeq 7.2 \pm 1.7$. Xiong et al. measured $C \simeq 3.67$, $\beta \simeq 3.8$ and $\gamma \simeq 8.3$ for this model~\cite{Xiong2021}, all of which are within our confidence intervals, so we replicate their finding that this model explains the data. (We note that the confidence intervals for the parameters derive from our sampling uncertainty.) On the other hand, the same authors report a poor fit to model~\eqref{eq:pureexpform}. In contrast, we find that setting $N \simeq 251.3 \pm 0.8$, $\beta \simeq -0.51 \pm 0.35$, and $\gamma \simeq -1.2 \pm 1.9$ in model~\eqref{eq:pureexpform} fits our data equally well. We~conclude that more will be needed to distinguish between~\eqref{eq:xdwform} and \eqref{eq:pureexpform}.

\begin{figure}[ht]
\hphantom{.}
\hfill
\begin{overpic}[width=0.45\textwidth]{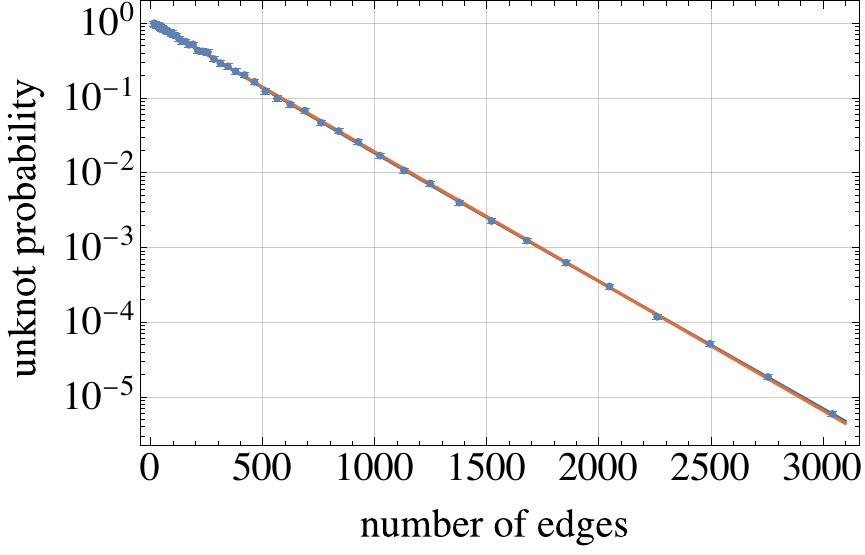}
\end{overpic}
\hfill
\begin{overpic}[width=0.45\textwidth]{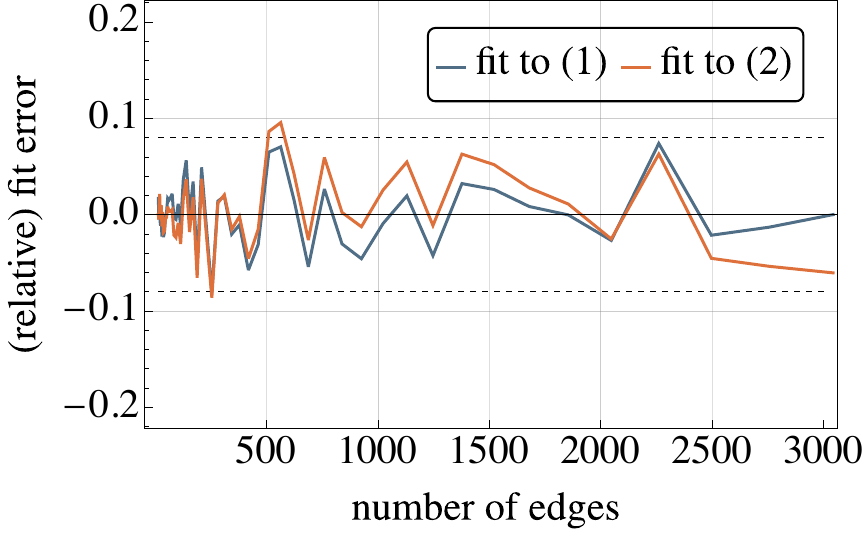}
\end{overpic}
\hfill
\hphantom{.}
\caption{\label{fig:fits} At left, we show a log plot of the probability that a given $n$-edge equilateral polygon has $\Delta_2$, $\Delta_3$, and $\Delta_4$ all equal to $1$ together with $95\%$ confidence intervals. Following~\cite{Xiong2021}, we use this as a proxy for the unknot probability.  We also show our best fits to~\eqref{eq:xdwform} and~\eqref{eq:pureexpform}. At right, we show the relative fit errors together with the measurement uncertainty (dotted lines). We see that both models fit the data extremely well.}
\end{figure}

\begin{table}[bp]

\begin{tabular}{l@{\hskip 0.125in}l@{\hskip 0.25in}l@{\hskip 0.125in}l@{\hskip 0.25in}l@{\hskip 0.125in}l}
 \toprule
$n$&$P_{0_1}(n)$&$n$&$P_{0_1}(n)$&$n$&$P_{0_1}(n)$\\
\midrule
$16$ & $0.99\pm 0.08$ & $95$ & $0.73\pm 0.06$ & $565$ & $0.098\pm 0.008$ \\
 $18$ & $0.99\pm 0.08$ & $105$ & $0.69\pm 0.06$ & $624$ & $0.082\pm 0.007$ \\
 $20$ & $0.96\pm 0.08$ & $116$ & $0.67\pm 0.05$ & $689$ & $0.067\pm 0.005$ \\
 $22$ & $0.95\pm 0.08$ & $128$ & $0.61\pm 0.05$ & $761$ & $0.046\pm 0.004$ \\
 $24$ & $0.95\pm 0.08$ & $141$ & $0.57\pm 0.05$ & $840$ & $0.0358\pm 0.0029$ \\
 $26$ & $0.95\pm 0.08$ & $156$ & $0.57\pm 0.05$ & $927$ & $0.0257\pm 0.0021$ \\
 $29$ & $0.93\pm 0.07$ & $172$ & $0.51\pm 0.04$ & $1024$ & $0.0168\pm 0.0013$ \\
 $32$ & $0.94\pm 0.08$ & $190$ & $0.52\pm 0.04$ & $1131$ & $0.0107\pm 0.0009$ \\
 $35$ & $0.91\pm 0.07$ & $210$ & $0.430\pm 0.034$ & $1248$ & $0.0071\pm 0.0006$ \\
 $39$ & $0.92\pm 0.07$ & $232$ & $0.418\pm 0.033$ & $1378$ & $0.00396\pm 0.00032$ \\
 $43$ & $0.90\pm 0.07$ & $256$ & $0.406\pm 0.032$ & $1522$ & $0.00225\pm 0.00018$ \\
 $48$ & $0.87\pm 0.07$ & $283$ & $0.329\pm 0.026$ & $1680$ & $0.00123\pm 0.00010$ \\
 $53$ & $0.85\pm 0.07$ & $312$ & $0.290\pm 0.023$ & $1855$ & $0.00062\pm 0.00005$ \\
 $58$ & $0.82\pm 0.07$ & $345$ & $0.264\pm 0.021$ & $2048$ & $0.000300\pm 0.000024$ \\
 $64$ & $0.81\pm 0.06$ & $380$ & $0.226\pm 0.018$ & $2261$ & $0.000118\pm 0.000009$ \\
 $71$ & $0.78\pm 0.06$ & $420$ & $0.201\pm 0.016$ & $2497$ & $0.000051\pm 0.000004$ \\
 $78$ & $0.76\pm 0.06$ & $464$ & $0.164\pm 0.013$ & $2756$ & $0.0000184\pm 0.0000015$ \\
 $86$ & $0.76\pm 0.06$ & $512$ & $0.122\pm 0.010$ & $3043$ & $(5.9\pm 0.5)\times 10^{-6}$ \\ \bottomrule
\end{tabular}

\caption{\label{tab:unknot probs}This table gives the probability of finding a random polygon with $\Delta_2$, $\Delta_3$ and $\Delta_4$ all equal to $1$, with error given by the $95\%$ confidence interval. We take this to be a proxy for the probability $P_{0_1}(n)$ of finding an unknot among random $n$-edge equilateral polygons. We expect that these numbers are highly accurate; checking each candidate unknot with SnapPy~\cite{SnapPy} revealed only 8 instances (of $32,\!400$) in which a nontrivial knot was mistaken for an unknot. We do not include this in the error bars shown.}

\end{table}

\section*{Acknowledgments}
We are very grateful for the On-Line Encyclopedia of Integer Sequences~\cite{oeis}, without which progress on this paper would have been much slower. Thanks to the anonymous referees for their thoughtful suggestions which substantially improved this paper, to Bertrand Duplantier for helping us think asymptotically, to Tetsuo Deguchi for a helpful discussion of probability models for knots, to Nathan Dunfield for suggesting that we check the unknot candidates with SnapPy. In addition, we would like to acknowledge the generous support of the National Science Foundation (DMS--2107700 to Shonkwiler), the Simons Foundation (\#524120 to Cantarella, \#709150 to Shonkwiler) and the Banff International Research Station (Workshop~24w5217, \emph{Knot Theory Informed by Random Models and Experimental Data}).

\newpage

\bibliography{tcrwpapers-special,extra-jason-refs,tcrwpapers,software}

\end{document}